\documentclass{article}

\usepackage{times}
\usepackage{helvet}
\usepackage{courier}
\usepackage{graphicx}
\usepackage{natbib}
\usepackage{caption}
\DeclareCaptionStyle{ruled}{labelfont=normalfont,labelsep=colon,strut=off}

\usepackage{amsthm}
\usepackage{amsmath}
\usepackage{amsfonts}
\usepackage{amssymb}
\usepackage{tikz}
\usetikzlibrary{shapes.geometric}
\usetikzlibrary{positioning}
\usepackage[T1]{fontenc}
\usepackage{dblfloatfix}
\usetikzlibrary{decorations.pathmorphing}

\newtheorem{theorem}{Theorem}
\newtheorem{lemma}[theorem]{Lemma}
\newtheorem*{lemma*}{Lemma ($*$)}
\newtheorem{proposition}[theorem]{Proposition}
\newtheorem{definition}{Definition}

\newtheorem{innercustomgeneric}{\customgenericname}
\providecommand{\customgenericname}{}
\newcommand{\newcustomtheorem}[2]{%
  \newenvironment{#1}[1]
  {%
   \renewcommand\customgenericname{#2}%
   \renewcommand\theinnercustomgeneric{##1}%
   \innercustomgeneric
  }
  {\endinnercustomgeneric}
}

\newcustomtheorem{theoremapp}{Theorem}
\newcustomtheorem{lemmaapp}{Lemma}
\newcustomtheorem{propositionapp}{Proposition}

\title{Closeness Centrality via the Condorcet Principle}
\author{Oskar Skibski\\\normalsize University of Warsaw}

\begin{document}

\maketitle

\begin{abstract}
We uncover a new relation between Closeness centrality and the Condorcet principle.
We define a Condorcet winner in a graph as a node that compared to any other node is closer to more nodes.
In other words, if we assume that nodes vote on a closer candidate, a Condorcet winner would win a two-candidate election against any other node in a plurality vote.
We show that Closeness centrality and its random-walk version, Random-Walk Closeness centrality, are the only classic centrality measures that are Condorcet consistent on trees, i.e., if a Condorcet winner exists, they rank it first.
While they are not Condorcet consistent in general graphs, we show that Closeness centrality satisfies the Condorcet Comparison property that states that out of two adjacent nodes, the one preferred by more nodes has higher centrality.
We show that Closeness centrality is the only regular distance-based centrality with such a property.
\end{abstract}

\section{Introduction}
\emph{Closeness centrality} proposed by \citet{Bavelas:1950} was one of the first methods in the literature designed to evaluate the position of a node in the network.
Defined as the inverse of the sum of distances to other nodes, it was originally used to study the role of identities in group collaboration.
This is when the centrality analysis originated and the term ``centrality'' has been coined as an indicator of the \emph{central} position in the network.

Over the last 70 years, network analysis has experienced tremendous growth during which the term ``centrality'' has been revisited and redefined multiple times.
Notably, Betweenness centrality and other medial centralities have been proposed to capture not the central position but the control over the communication between other nodes~\cite{Freeman:1977,Freeman:etal:1991}.
Furthermore, PageRank and other feedback centralities focused more on the number, and the importance of the directly connected nodes rather than on being close to all the nodes~\cite{Page:etal:1999,Bonacich:1987}.
In the last two decades, dozens of new centrality concepts have been proposed, drifting even further from the origin of their name~\cite{Jalili:etal:2015}.
Nevertheless, Closeness centrality has remained one of the key and the most commonly used measures of centrality.

Many works in the literature analyzed various properties of Closeness centrality, both from the theoretical~\cite{Brandes:etal:2016,Borgatti:Everett:2006}
and empirical perspective~\cite{Powell:etal:1996,Rowley:1997}.
Interestingly, in this paper, we uncover a new, not yet identified relation between Closeness centrality and social choice theory.

We interpret a graph as an election in which nodes are both candidates and voters.
The voters' preferences are defined based on their distance to other nodes.
Specifically, out of two nodes (candidates), the node (voter) prefers the closer one.
Now, a Condorcet winner in a graph is a node that compared to any other node is closer to more nodes.
In other words, a Condorcet winner is a node that would win a two-candidate election against any other node in a plurality vote.


We obtain several results for this setting.
We begin with the analysis of trees and prove that Closeness centrality is Condorcet consistent in trees, i.e., a Condorcet winner, if it exists, is ranked first by Closeness centrality.
Out of standard centrality measures, only Closeness and its random-walk version Random-Walk Closeness have this property.

Next, we focus on arbitrary graphs.
As it turns out, Closeness centrality and other centralities based solely on the distances to other nodes are not Condorcet consistent.
However, we show that Closeness centrality satisfies the property that we name \emph{Condorcet Comparison}: out of two adjacent nodes, the one preferred by more nodes has higher centrality.
We show that Closeness centrality is the only regular distance-based centrality with such a property.
More precisely, we show that if a regular distance-based centrality satisfies Condorcet Comparison, then it returns the same ranking as Closeness.

Our work sheds new light on centrality analysis from the social choice perspective.
In particular, our results contribute to the discussion on extending Closeness centrality to arbitrary graphs. 
Specifically, Closeness centrality is well-defined only for connected graphs; hence several distance-based centralities, such as Harmonic and Decay centralities, are often advocated as good alternatives. 
Our analysis, however, provides evidence that, more often than not, they behave significantly different than Closeness centrality.

While we focus only on Closeness centrality and Condorcet principle, our work opens up many interesting future directions.
We discuss them in Conclusions.


\section{Preliminaries}

In this paper, we consider undirected, unweighted graphs.

A \emph{graph} is a pair $G = (V,E)$ where $V$ is the set of nodes and $E$ is the set of edges, i.e., pairs $\{u,v\} \subseteq V$.
We assume that $|V|=n$.

A (simple) \emph{path} is a sequence of different nodes $(v_0, \dots, v_k)$ such that $\{v_i, v_{i+1}\} \in E$ for every $i \in \{0,\dots, k-1\}$.
We say that this path \emph{starts} in node $v_0$, ends in $v_k$ and has length $k$.
The \emph{distance} between $u$ and $v$, denoted by $d(u,v)$, is the length of a shortest path that starts in $u$ and ends in $v$.
We assume that $d(v,v)=0$ for every node $v$ and $d(u,v) = +\infty$ if there is no path between $u$ and $v$.

Node $u$ is called a \emph{neighbor} of node $v$ if there is an edge between them: $\{u,v\} \in E$.
The number of neighbors of $v$ is called its \emph{degree} and denoted by $\deg(v)$.

The \emph{list of distances} of node $v$ is a list of positive natural numbers $A(v) = (a_1,\dots,a_k)$ such that $a_i = |\{u \in V : d(u,v) = i\}|$ and $k = \max_u d(u,v)$.
To put it in words, $a_i$ is the number of nodes at distance $i$.
For example, for node $v$ from Figure~\ref{fig:preliminaries} we have $A(v) = (2,2,4,4)$.

A graph is a \emph{tree} if there exists exactly one path between any two nodes.
If there is at least one path between any two nodes, then we say that the graph is \emph{connected}.
A \emph{connected component} of a graph is a subset of nodes that contains all neighbors of all nodes that belong to it.
The set of connected components of a graph is denoted by $K(G)$.
In particular, if $G$ is connected, then $K(G) = \{V\}$.

An edge $\{u,v\}$ is a \emph{bridge} in a connected graph $G$ if its removal makes the graph disconnected: $|K(G-e)| > 1$.
Here, $G-e$ denotes the graph $(V, E \setminus \{e\})$.
A \emph{subgraph induced} by $S \subseteq V$ is denoted by $G[S]$ and defined as follows:
\[ G[S] = (S, E[S]) = (S, \{\{u,v\} \in E : u,v \in S\}). \]

A \emph{random walk} on a graph $G = (V,E)$ is a random sequence of nodes $(v_0, v_1, \dots)$ characterized by a starting node ($v_0$) and the transition probability:
\[ \mathrm{Pr}[v_{t+1} = w \mid v_t = u] = \begin{cases}
1/\deg(u) & \mbox{ if $\{u,w\} \in E$,} \\
0 & \mbox{otherwise.}
\end{cases} \]
To put it in words, in every step, the random walk picks a random edge incident to the node it is in and follows it to the destination of that edge.

For two nodes $u,v \in V$, the \emph{hitting time} $H(u,v)$ is the expected number of steps of the random walk starting in node $u$ before node $v$ is visited.
We assume $H(v,v) = 0$ for every $v \in V$.
Similarly, the \emph{expected return time} of node $v$ is the expected number of steps of the random walk starting in node $v$ before node $v$ is visited again.
It is known that the expected return time of node $v$ in a connected graph $G = (V,E)$ equals $2|E|/\deg(v)$~\cite{Lovasz:1993}.

\subsection{Centrality Measures}
A \emph{centrality measure} is a function that assesses the importance of nodes in a graph. 
Specifically, centrality measure $F$ in every graph $G = (V,E)$ to every node $v \in V$ assigns a real value, denoted by $F_v(G)$.
In this paper, we focus on Closeness centrality \cite{Bavelas:1950}.

\begin{definition}
\emph{Closeness centrality} of node $v$ in a connected graph $G = (V,E)$ is defined as follows:
\[ C_v(G) = \Bigg(\sum_{u \in V \setminus \{v\}} d(u,v)\Bigg)^{-1}. \]
\end{definition}
Closeness centrality is sometimes normalized by multiplying by $n-1$. 
Written like this, it is equal to the inverse of the average distance to other nodes in a graph.

A centrality measure $F$ is \emph{distance-based} if it depends only on the list of distances of a node: $F_v(G) = f(A(v))$ for some function $f$.
Closeness centrality is distance-based with the function: $f(a_1,\dots,a_k) = (\sum_{i=1}^{k} a_i i)^{-1}$.
For other distance-based centralities we have:
\begin{itemize}
\item \emph{Degree centrality}~\cite{Freeman:1977}: $f(a_1,\dots,a_k) = a_1$;
\item \emph{Harmonic centrality}~\cite{Rochat:2009}: $f(a_1,\dots,a_k) = \sum_{i=1}^k (a_i/i)$; and
\item \emph{Decay centrality}~\cite{Jackson:2008}: $f(a_1,\dots,a_k) = \sum_{i=1}^k (a_i \delta^i)$ for some $\delta \in (0,1)$.
\end{itemize}
Harmonic and Decay centralities are well-defined also for disconnected graphs under the assumption that $1/\infty = \delta^{\infty} = 0$ for every $\delta \in (0,1)$.
Hence, they are often referred to as an extension of Closeness centrality to arbitrary graphs.

Closeness centrality focuses on shortest paths.
An alternative based on the random walk was proposed by \citet{White:Smyth:2003}. 
Here, instead of the distance between nodes, we look at the expected number of steps performed by the random walk, i.e., the hitting time.
\begin{definition}
\emph{Random-Walk Closeness (RW-Closeness) centrality} of node $v$ in a connected graph $G = (V,E)$ is:
\[ RWC_v(G) = \Bigg(\sum_{u \in V \setminus \{v\}} H(u,v) \Bigg)^{-1}. \]
\end{definition}
This centrality was originally named \emph{Markov centrality}, but we will use the more informative name proposed by \citet{Brandes:Erlebach:2005}.
We also note that the expected return time of node $v$ is sometimes added to the sum which spoils most of the properties that we will discuss.

Let us denote by $\mathrm{Top}_F(G)$ the set of nodes which have the highest values in graph $G$ according to centrality $F$: 
$\mathrm{Top}_F(G) = \{v \in V : F_v(G) \ge F_u(G) \mbox{ for every }u \in V\}$.
While all centrality measures listed in this section are similar, they often differ in the nodes they rank first.
An example of that is presented in Figure~\ref{fig:preliminaries}.

\begin{figure}[b]
\centering
\begin{tikzpicture}[x=8cm,y=6cm] 
  \tikzset{     
    e4c node/.style={circle,draw,minimum size=0.50cm,inner sep=0}, 
    e4c edge/.style={sloped,above,font=\footnotesize}
  }
  
  \node[e4c node,minimum size=0.58cm] (6) at (0.62, 0.00) {$v$};
 
  \node[e4c node] (2) at (0.00, 0.00) {}; 
  \node[e4c node] (1) at (0.10, 0.11) {}; 
  \node[e4c node] (3) at (0.10, -0.11) {}; 
  \node[e4c node] (4) at (0.20, 0.00) {$x$}; 
  \node[e4c node] (45) at (0.34, 0.00) {}; 
  \node[e4c node] (5) at (0.48, 0.00) {$u$}; 
  \node[e4c node] (6) at (0.62, 0.00) {$v$}; 
  \node[e4c node] (67) at (0.76, 0.00) {$w$}; 
  \node[e4c node] (7) at (0.90, 0.00) {$y$}; 
  \node[e4c node] (8) at (0.90, 0.11) {}; 
  \node[e4c node] (9) at (0.90, -0.11) {}; 
  \node[e4c node] (10) at (1.04, 0.00) {}; 
  \node[e4c node] (11) at (1.18, 0.00) {}; 

  \path[draw,thick]
  (1) edge[e4c edge]  (2)
  (3) edge[e4c edge]  (1)
  (2) edge[e4c edge]  (3)
  (1) edge[e4c edge]  (4)
  (2) edge[e4c edge]  (4)
  (3) edge[e4c edge]  (4)
  (45) edge[e4c edge]  (5)
  (9) edge[e4c edge]  (7)
  (11) edge[e4c edge]  (10)
  (67) edge[e4c edge]  (7)
  (6) edge[e4c edge]  (67)
  (5) edge[e4c edge]  (6)
  (7) edge[e4c edge]  (10)
  (7) edge[e4c edge]  (8)
  (45) edge[e4c edge]  (4)
  ;
\end{tikzpicture}
\caption{Nodes $u,v,w,y$ are ranked highest by RW-Closeness, Closeness, Decay ($\delta=0.8$) and Harmonic centralities, respectively. Nodes $x$ and $y$ have the highest degree.}
\label{fig:preliminaries}
\end{figure}
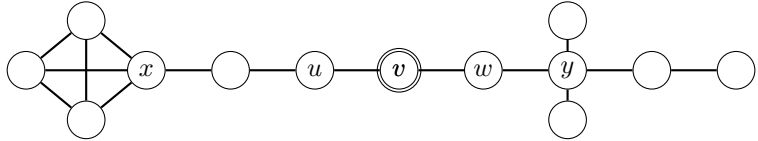

\section{Condorcet in Graphs}
In this paper, we will interpret a graph as an \emph{election} in which nodes are both candidates and voters.

The references of voters will be defined based on their distance to other nodes.
We say that node $w$ prefers node $u$ to node $v$ if $d(w,u) < d(w,v)$.
The number of nodes that prefers $u$ to $v$ will be denoted by $Net(u,v)$:
\[ Net(u,v) = |\{w \in V : d(w,u) < d(w,v)\}| \]
We will write $u \succ v$ if $Net(u,v) > Net(v,u)$ and $u \sim v$ if $Net(u,v) = Net(v,u)$.
Also, $u \succeq v$ if $u \succ v$ or $u \sim v$

A node $u$ is a \emph{Condorcet winner} if for every node $v$ ($v \neq u$) the number of nodes that prefers $u$ to $v$ is larger than the number of nodes that prefers $v$ to $u$, i.e., $u \succ v$.
In other words, node $u$ has more supporters in the head-to-head comparison with any other node.

We will denote a Condorcet winner on figures by the double line.
For example, in Figure~\ref{fig:preliminaries} node $v$ is a Condorcet winner: it divides the graph into two parts with $6$ nodes, hence for any other node $s$ it holds $Net(v,s) \ge 7 > Net(s,v)$.

In general, a Condorcet winner may not exist.
A \emph{Condorcet cycle} is a sequence $(v_1,\dots,v_k)$ such that $v_i \succ v_{i+1}$ for every $i \in \{1,\dots,k-1\}$ and $v_k \succ v_1$.
See Figure~\ref{fig:cycle} for an example.

Finally, we will say that a centrality measure is \emph{Condorcet consistent} if a Condorcet winner has the highest centrality in a graph if such a winner exists at all.

\begin{definition} (Condorcet Consistency)
A centrality measure $F$ is \emph{Condorcet consistent} if for every graph $G$ with a Condorcet winner $u$ it holds $\mathrm{Top}_F(G) = \{u\}$.
\end{definition}

To date, only \citet{Telek:2017} considered this model.
He proved that in trees Condorcet winner often exists (in particular, if there is an odd number of nodes). However, if a Condorcet winner does not exist, then there exist two nodes $u,v$ such that $\{u,v\} \in E$, $u \sim v$ and $u \succ w$, $v \succ w$ for every $w \in V \setminus \{u,v\}$.
We will call these two nodes \emph{weak Condorcet winners}.

\citet{Telek:2017} also characterized two other simple classes of graphs on which Condorcet winner exists. 
However, the only link to centrality measures in his work is the observation that Betweenness and Eigenvector centralities are not Condorcet consistent.

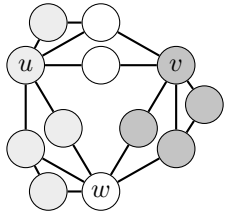
\begin{figure}[b]
\centering
\begin{tikzpicture}[x=2.5cm,y=2.1cm] 
  \tikzset{     
    e4c node/.style={circle,draw,minimum size=0.50cm,inner sep=0}, 
    e4c edge/.style={sloped,above,font=\footnotesize}
  }
  
  \node[e4c node,fill=black!23] (1) at (-0.20, 1.00) {$v$}; 
  \node[e4c node,fill=black!07] (2) at (-1.00, 1.00) {$u$}; 
  \node[e4c node] (3) at (-0.60, 0.20) {$w$}; 
  \node[e4c node] (4) at (-0.60, 1.00) {}; 
  \node[e4c node,fill=black!07] (5) at (-0.80, 0.60) {}; 
  \node[e4c node,fill=black!23] (6) at (-0.40, 0.60) {}; 
  \node[e4c node] (7) at (-0.60, 1.27) {}; 
  \node[e4c node,fill=black!07] (8) at (-0.88, 1.27) {}; 
  \node[e4c node,fill=black!07] (9) at (-1.00, 0.47) {}; 
  \node[e4c node,fill=black!07] (10) at (-0.88, 0.20) {}; 
  \node[e4c node,fill=black!23] (11) at (-0.20, 0.47) {}; 
  \node[e4c node,fill=black!23] (12) at (-0.05, 0.75) {}; 


  \path[draw,thick]
  (1) edge[e4c edge]  (4)
  (4) edge[e4c edge]  (2)
  (2) edge[e4c edge]  (5)
  (5) edge[e4c edge]  (3)
  (3) edge[e4c edge]  (6)
  (6) edge[e4c edge]  (1)
  (1) edge[e4c edge]  (7)
  (7) edge[e4c edge]  (8)
  (8) edge[e4c edge]  (2)
  (7) edge[e4c edge]  (2)
  (2) edge[e4c edge]  (9)
  (9) edge[e4c edge]  (10)
  (10) edge[e4c edge]  (3)
  (9) edge[e4c edge]  (3)
  (3) edge[e4c edge]  (11)
  (11) edge[e4c edge]  (12)
  (12) edge[e4c edge]  (1)
  (11) edge[e4c edge]  (1)
  ;
\end{tikzpicture}
\caption{An example of a Condorcet cycle $(u,v,w)$. Light gray nodes are closer to $u$ than to $v$ ($Net(u,v) = 5$) and dark gray nodes are closer to $v$ than to $u$ ($Net(v,u) = 4$). By symmetry, we get that $u \succ v$, $v \succ w$ and $w \succ u$.}
\label{fig:cycle}
\end{figure}

\section{Closeness in Trees}

We begin with the analysis of trees.

The crucial role in our analysis will be played by the axiom that we term \emph{Condorcet Comparison}.
This axiom states that the ranking between two adjacent nodes is determined by their head-to-head comparison.
In this section, we present its version restricted to trees.

\begin{definition} (Condorcet Comparison in Trees (CCT)) 
For every tree $G = (V,E)$ and edge $\{u,v\} \in E$:
\[ u \succeq v \Leftrightarrow F_u(G) \ge F_v(G). \]
\end{definition}

Note that this definition implies also that $u \succ v$ iff $F_u(G) > F_v(G)$ and $u \sim v$ iff $F_u(G) = F_v(G)$.

CCT is related to the \emph{Bridge} axiom proposed by \citet{Skibski:Sosnowska:2018:distance}.
Bridge states that if $\{u,v\}$ is a bridge in a graph connecting two sets of nodes $S_u$ and $S_v$ s.t. $u \in S_u$ and $v \in S_v$, then the node from the larger set has larger centrality.
Now, if we consider the graph as an election, we see that all nodes from $S_u$ prefers $u$ to $v$ and all nodes from $S_v$ prefer $v$ to $u$.
Hence, the condition $|S_u| \ge |S_v|$ is equivalent to $u \succeq v$ and CCT is equivalent to Bridge for trees.

It is easy to verify that Closeness centrality satisfies CCT.
Following the above notation, assume $K(G-\{u,v\}) = \{S_u,S_v\}$ s.t. $u \in S_u$ and $v \in S_v$.
Node $u$ has distance larger by one than $v$ to nodes from $S_v$ and $v$ has distance larger by one than $u$ to nodes from $S_u$. 
Hence:
\begin{equation}\label{eq:closeness_in_trees}
C_v^{-1}(G)-C_u^{-1}(G) = |S_u|-|S_v|.
\end{equation}

The following lemma shows that satisfying CCT is enough to imply Condorcet consistency on trees.
Specifically, a centrality measure that satisfies CCT ranks first a Condorcet winner or, if it does not exist, the two weak Condorcet winners.

\begin{lemma}\label{lemma:trees_axiom_condorcet}
If a centrality measure $F$ satisfies CCT, then for every tree $G$ and nodes $u,v \in \mathrm{Top}_F(G)$ and $w \not \in \mathrm{Top}_F(G)$ it holds $u \sim v$ and $u \succ w$.
Hence, $F$ is Condorcet consistent on trees.
\end{lemma}
\begin{proof}
First, let $(u,v,w)$ be a path in a tree and assume $u \succ v$. 
We have $Net(v,u) \ge Net(w,u) \ge Net(w,v)$.
Analogously, $Net(u,v) \le Net(u,w) \le Net(v,w)$.
Hence, from $Net(u,v) > Net(v,u)$ we get $Net(v,w) > Net(w,v)$ which means $v \succ w$.

Now, fix a tree $G$ and assume $F$ satisfies CCT.
Assume there exists a Condorcet winner $u$.
Let $(u,v_1,\dots,v_k,w)$ be a path to an arbitrary node $w$. 
From the above observation we know that $u \succ v_1$ implies $v_1 \succ v_2$.
Also, $v_1 \succ v_2$ implies $v_2 \succ v_3$ and so on.
Eventually, we get $v_i \succ v_{i+1}$ for every $i \in \{1,\dots,k-1\}$ and $v_k \succ w$.
If $F$ satisfies CCT, then this implies $F_u(G) > \dots > F_w(G)$ and $\mathrm{Top}_F(G) = \{u\}$.

Now, assume there is no Condorcet winner and let $u,v$ be weak Condorcet winners.
From CCT and the fact that $u \sim v$ we know that $F_u(G) = F_v(G)$.
Now, for every node $w \in V \setminus \{u,v\}$ there exists a path from $u$ or $v$ that does not go through the other node.
Hence, the reasoning from the first part of the proof can be repeated to show that $F_u(G) > F_w(G)$. Eventually we get that $\mathrm{Top}_F(G) = \{u,v\}$ which concludes the proof.
\end{proof}

This leads us to the following result.

\begin{theorem}
Closeness centrality is Condorcet consistent on trees.
\end{theorem}
\begin{proof}
Directly from Lemma~\ref{lemma:trees_axiom_condorcet} and the fact that Closeness centrality satisfies CCT.
\end{proof}

Hardly any other centrality measure is Condorcet consistent on trees.
Consider the graph from Figure~\ref{fig:trees_counterexample}.
Here, node $v$ is a Condorcet winner.
However, most classic centrality measures, including other distance-based centralities (Degree, Harmonic, Decay), medial centralities (Betweenness, Stress), and feedback centralities (PageRank, Eigenvector, Katz), all rank node $u$ at the top.

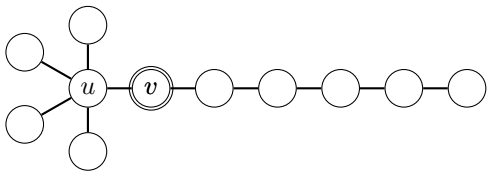
\begin{figure}[t]
\centering
\begin{tikzpicture}[x=6cm,y=6cm] 
  \tikzset{     
    e4c node/.style={circle,draw,minimum size=0.5cm,inner sep=0}, 
    e4c edge/.style={sloped,above,font=\footnotesize}
  }
  
  \node[e4c node,minimum size=0.58cm] (1) at (0.28,0.0) {$v$};
  
  \node[e4c node] (0) at (0.98, 0.0) {}; 
  \node[e4c node] (1) at (0.84, 0.0) {}; 
  \node[e4c node] (2) at (0.70, 0.0) {}; 
  \node[e4c node] (3) at (0.56, 0.0) {}; 
  \node[e4c node] (4) at (0.42, 0.0) {}; 
  \node[e4c node] (5) at (0.28, 0.0) {$v$}; 
  \node[e4c node] (6) at (0.14, 0.0) {$u$}; 
  \node[e4c node] (7) at (0.14, -0.14) {}; 
  \node[e4c node] (8) at (0.14, +0.14) {}; 
  \node[e4c node] (9) at (0.00, +0.08) {}; 
  \node[e4c node] (10) at (0.00, -0.08) {}; 

  \path[draw,thick]
  (6) edge[e4c edge]  (9)
  (6) edge[e4c edge]  (8)
  (6) edge[e4c edge]  (7)
  (5) edge[e4c edge]  (6)
  (4) edge[e4c edge]  (5)
  (3) edge[e4c edge]  (4)
  (2) edge[e4c edge]  (3)
  (1) edge[e4c edge]  (2)
  (0) edge[e4c edge]  (1)
  (6) edge[e4c edge]  (10)
  ;
\end{tikzpicture}
\caption{Node $v$ is a Condorcet winner, while node $u$ is ranked first according to the majority of classic centralities.}
\label{fig:trees_counterexample}
\end{figure}

A notable exception is RW-Closeness centrality which, as we will show, is also Condorcet consistent.


\begin{theorem}\label{theorem:trees_rwc_consistent}
RW-Closeness centrality is Condorcet consistent on trees.
\end{theorem}
\begin{proof}
Based on Lemma~\ref{lemma:trees_axiom_condorcet} it is enough to prove that RW-Closeness satisfies CCT.
Fix $\{u,v\} \in E$ and let $K(G-\{u,v\}) = \{S_u, S_v\}$ s.t. $u \in S_u$, $v \in S_v$.
Consider the random walk starting in $w \in S_u$. 
Before the random walk reaches $v$ it must reach $u$ first. 
Hence, $H(w,v) = H(w,u) + H(u,v)$ and we get:
\[ RWC_v^{-1}(G) =  \sum_{w \in S_u} (H(w,u) + H(u,v)) + \sum_{w \in S_v} H(w,v). \]

Now, let us compute the hitting time $H(u,v)$.
Since $H(u,v)$ concerns the first time the random walk enters node $v$, it does not depend on nodes $S_v \setminus \{v\}$.
Hence, we can concentrate on the induced subgraph $G[S_u \cup \{v\}]$. 
Consider the expected return time of node $v$ in this subgraph.
Since $G[S_u \cup \{v\}]$ is a tree with $|S_u|$ edges and $v$ is a leaf, the expected return time of node $v$ in graph $G[S_u \cup \{v\}]$ equals $2|S_u|$.
Now, if the random walk starts in $u$, it will reach node $v$ one step earlier.
Hence, we get that $H(u,v) = 2|S_u|-1$.

Overall, we get that:
\begin{equation*}
RWC_v^{-1}(G) - RWC_u^{-1}(G) =
|S_u|(2|S_u|-1) - |S_v|(2|S_v|-1).
\end{equation*}
From the fact that $|S_u|^2-|S_v|^2 = (|S_u|-|S_v|)(|S_u|+|S_v|)$ and $|S_u|+|S_v| = n$ we have:
\begin{equation}\label{eq:proof_rwc}
RWC_v^{-1}(G) - RWC_u^{-1}(G) = (|S_u|-|S_v|)(2n-1)
\end{equation}
which implies the thesis.
\end{proof}

The fact that Closeness and RW-Closeness centralities both satisfy CCT does not imply they impose the same ranking of nodes.
In particular, consider the following centrality:
\begin{equation}\label{eq:centrality_zero_to_leafs}
X_v(G) = \begin{cases}
C_v(G) & \mbox{ if }\deg(v) > 1,\\
0 & \mbox{otherwise.}
\end{cases}
\end{equation}
Clearly, this centrality satisfies CCT, but results in a different ranking than Closeness centrality.
It is also distance-based, which means that Closeness centrality is not the only distance-based centrality that satisfies CCT.

While CCT does not imply a unique ranking, our analysis from Theorem~\ref{theorem:trees_rwc_consistent} implies that indeed Closeness and RW-Closeness centralities rank nodes in the same way.

\begin{theorem}\label{theorem:trees_closeness_rwc}
Closeness centrality and RW-Closeness centrality impose the same ranking of nodes in a tree.
\end{theorem}
\begin{proof} 
Let $(v_0, v_1, \dots, v_k)$ be an arbitrary path in a tree $G = (V,E)$.
Fix $i \in \{0,\dots,k-1\}$.
From Equations~\eqref{eq:closeness_in_trees} and \eqref{eq:proof_rwc} we get that:
\[ \frac{RWC_{v_i}^{-1}(G)-RWC_{v_{i+1}}^{-1}(G) }{(2n-1)} = C_{v_i}^{-1}(G)-C_{v_{i+1}}^{-1}(G). \]
Summing over all $i \in \{0,\dots,k-1\}$ we get that $C_{v_0}(G) \le C_{v_k}(G)$ if and only if $RWC_{v_0}(G) \le RWC_{v_k}(G)$.
\end{proof}

\subsection{Condorcet Comparison for non-adjacent nodes}

CCT applies only to nodes connected by an edge. 
In this section, we ask the question of whether this assumption can be relaxed.
First, consider the following version of CCT.

\begin{definition} (General CCT) For every tree $G = (V,E)$ and two nodes $u,v \in V$:
$u \succeq v \Leftrightarrow F_u(G) \ge F_v(G)$.
\end{definition}

It is easy to check that Closeness centrality does not satisfy General CCT.
Roughly speaking, this is because Closeness centrality when it compares two nodes cares not only about the number of nodes closer to each of them but also about how much closer these nodes are.
A counterexample is presented in Figure~\ref{fig:trees_closeness_general_comparison}.
Except for the nodes on the path between $u$ and $v$, we have two nodes $w$ with $d(w,u) = d(w,v)-1$ and one node with $d(w,u)=d(w,v)+3$. 

\begin{figure}[t]
\centering
\begin{tikzpicture}[x=4cm,y=4cm] 
  \tikzset{     
    e4c node/.style={circle,draw,minimum size=0.5cm,inner sep=0}, 
    e4c edge/.style={sloped,above,font=\footnotesize}
  }
  \node[e4c node,minimum size=0.58cm] (5) at (0.25, 0.20) {}; 
   
  \node[e4c node] (1) at (0.00, 0.20) {$u$}; 
  \node[e4c node] (3) at (0.25, 0.40) {}; 
  \node[e4c node] (4) at (0.25, 0.00) {}; 
  \node[e4c node] (5) at (0.25, 0.20) {}; 
  \node[e4c node] (6) at (0.50, 0.20) {}; 
  \node[e4c node] (7) at (0.75, 0.20) {$v$}; 
  \node[e4c node] (8) at (1.00, 0.20) {};

  \path[draw,thick]
  (1) edge[e4c edge]  (5)
  (3) edge[e4c edge]  (5)
  (4) edge[e4c edge]  (5)
  (5) edge[e4c edge]  (6)
  (6) edge[e4c edge]  (7)
  (7) edge[e4c edge]  (8)
  ;
\end{tikzpicture}
\caption{We have $u \succ v$, but $C_u(G) < C_v(G)$.}
\label{fig:trees_closeness_general_comparison}
\end{figure}
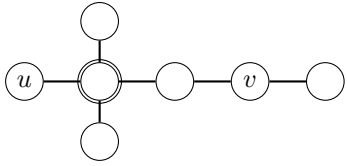

We note, however, a stronger result holds: General CCT cannot be satisfied by any centrality measure.
This is because $\succeq$ relation is not transitive.
To see that, consider the graph from Figure~\ref{fig:trees_new_measure}. 
We have $u \sim v$ and $v \sim w$, but $w \succ u$ (hence, $u \succeq v \succeq w$, but $u  \not \succeq w$).
General CCT for this graph implies $F_u(G) = F_v(G) = F_w(G)$, but at the same time $F_w(G) > F_u(G)$.

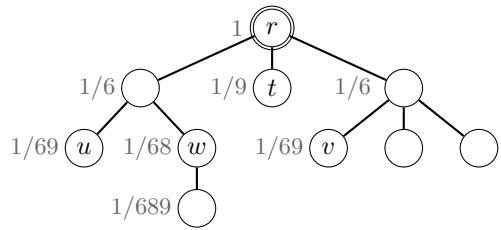
\begin{figure}[b]
\centering
\begin{tikzpicture}[x=2.5cm,y=2cm] 
  \tikzset{     
    e4c node/.style={circle,draw,minimum size=0.5cm,inner sep=0}, 
    e4c edge/.style={sloped,above,font=\footnotesize}
  }
  
  \node[e4c node,minimum size=0.58cm] (1) at (0.60,2.00) {};

  \node[e4c node] (1) at (0.60, 2.00) {$r$}; 
  \node[e4c node] (2) at (-0.10, 1.60) {}; 
  \node[e4c node] (3) at (1.30, 1.60) {}; 
  \node[e4c node] (4) at (0.60, 1.60) {$t$}; 
  \node[e4c node] (5) at (-0.40, 1.20) {$u$}; 
  \node[e4c node] (6) at (0.20, 1.20) {$w$}; 
  \node[e4c node] (7) at (0.20, 0.80) {}; 
  \node[e4c node] (8) at (0.90, 1.20) {$v$}; 
  \node[e4c node] (9) at (1.30, 1.20) {}; 
  \node[e4c node] (10) at (1.70, 1.20) {}; 

  \node[left=+0.00cm of 1,color=black!60] {\small $1$};
  \node[left=-0.05cm of 2,color=black!60] {\small $1/6$};
  \node[left=+0.05cm of 3,color=black!60] {\small $1/6$};
  \node[left=-0.05cm of 4,color=black!60] {\small $1/9$};
  \node[left=-0.05cm of 5,color=black!60] {\small $1/69$};
  \node[left=-0.05cm of 6,color=black!60] {\small $1/68$};
  \node[left=-0.05cm of 7,color=black!60] {\small $1/689$};
  \node[left=-0.05cm of 8,color=black!60] {\small $1/69$};

  \path[draw,thick]
  (1) edge[e4c edge]  (2)
  (1) edge[e4c edge]  (3)
  (1) edge[e4c edge]  (4)
  (2) edge[e4c edge]  (5)
  (2) edge[e4c edge]  (6)
  (6) edge[e4c edge]  (7)
  (3) edge[e4c edge]  (8)
  (3) edge[e4c edge]  (10)
  (3) edge[e4c edge]  (9)
  ;
\end{tikzpicture}
\caption{We have $u \sim v$, $v \sim w$, but $w \succ u$. The values next to nodes are computed using the measure defined in Proposition~\ref{proposition:trees_measure_exists}.}
\label{fig:trees_new_measure}
\end{figure}

From the above discussion, it is clear that there is no centrality measure that assigns equal centralities to every two nodes that tie in a head-to-head comparison.
Hence, let us consider a weaker property that focuses solely on the strict relation.
\begin{definition} (Weak General CCT) For every tree $G = (V,E)$ and two nodes $u,v \in V$:
$u \succ v \Rightarrow F_u(G) > F_v(G)$.
\end{definition}

This property is also not satisfied by Closeness centrality, as evident from Figure~\ref{fig:trees_closeness_general_comparison}.

Let us analyze in more detail $\succeq$ relation.
We begin with the following observation that states that to compare two nodes, it is enough to compare their neighbors on the path between them.

\begin{lemma}\label{lemma:moving_closer}
Let $(u,u',\dots,v',v)$ be a path in a tree of length at least 3.  
Then, $u \succeq v$ iff $u' \succeq v'$.
\end{lemma}
\begin{proof}
Take an arbitrary node $w \in V$. There are three cases that we consider separately:
\begin{itemize}
\item If the path from $v$ to $w$ goes through $u$, then the path from $v'$ to $w$ goes through $u'$. Hence, $d(w,u) < d(w,v)$ and $d(w,u') < d(w,v')$.
\item Analogously, if the path from $u$ to $w$ goes through $v$, then $d(w,u) > d(w,v)$ and $d(w,u') > d(w,v')$.
\item Assume otherwise. Then $d(w,u') = d(w,u) - 1$ and $d(w,v') = d(w,v) - 1$. Hence, $d(w,u) - d(w,v) = d(w,u') - d(w,v')$.
\end{itemize}
In result, we get that $N(u,v) = Net(u',v')$ and symmetrically $N(v,u) = Net(v',u')$.
This concludes the proof.
\end{proof}

Based on Lemma~\ref{lemma:moving_closer} we show that $\succeq$ relation depends on two factors: distance of nodes to the Condorcet winner(s) and sizes of subtrees if we root the tree in compared nodes.

Specifically, let us define the \emph{level} of node $v$, denoted by $l(v)$, as the distance to the Condorcet winner or the closer weak Condorcet winner.
Moreover, let $T_u^v$ be the set of nodes to which the path from $v$ goes through $u$:
\[ T_u^v = \{w \in V : d(w,v) = d(w,u) + d(u,v) \} \]
In other words, if we root the whole tree at $v$, then $T_u^v$ consists of nodes from a subtree rooted at $u$.

\begin{lemma}\label{lemma:relation_dist}
For every tree $G$, every two nodes $u,v$ it holds: $u \succeq v$ if and only if 
\[ (l(u) < l(v)) \lor (l(u) = l(v) \land |T^u_w| \le |T^v_w|), \]
where $w$ is the middle node (or the further middle node if there are two) on the path between $u$ and $v$, i.e., the unique node s.t. $d(u,w) = \lceil d(u,v)/2 \rceil$ and $d(w,v) = \lfloor d(u,v)/2 \rfloor$.
\end{lemma}
\begin{proof}
Take two nodes $u,v$.
Without loss of generality assume $l(u) \le l(v)$.
If $l(u) < l(v)$, then by applying Lemma~\ref{lemma:moving_closer} we get that $u \succeq v$ iff $u' \succeq v'$ where $u'$ belongs to the path from $v'$ to the Condorcet winner(s).
Hence, $u \succ v$.

Now, assume $l(u) = l(v)$ and $2 \nmid d(u,v)$.
This means that the path from $u$ to $v$ goes through two weak Condorcet winners $r,r'$ in that order.
Lemma~\ref{lemma:moving_closer} (applied $l(u)$ times) and $r \sim r'$ implies $u \sim v$.
This agrees with the lemma statement, as $w = r'$ and $|T^u_{r'}| \le |T^v_{r'}|$ is satisfied.

It remains to consider the case when $l(u) = l(v)$ and $2 \mid d(u,v)$.
Let $w$ be the middle node on the path between $u$ and $v$.
Let $u'$ and $v'$ be the neighbors of $w$ that belong to paths between $u$ and $w$ and $v$ and $w$, respectively.
From Lemma~\ref{lemma:moving_closer} we get that $u \succeq v$ iff $u' \succeq v'$.
We have:
\[ Net(u',v') = |T^w_{u'}| = n - |T^u_w|,\]
\[ Net(v',u') = |T^w_{v'}| = n - |T^v_w|. \]
Hence, $Net(u',v') - Net(v',u') = |T^v_w| - |T^u_w|$ which implies $u \succeq v$ iff $|T^u_w| \le |T^v_w|$.
\end{proof}

Based on Lemma~\ref{lemma:relation_dist} it is easy to show that $\succ$ relation is in fact transitive.

\begin{proposition}\label{proposition:trees_no_cycle}
In a tree, if $u \succ v$ and $v \succ w$, then $u \succ w$.
\end{proposition}
\begin{proof}
From Lemma~\ref{lemma:relation_dist} we know that $l(u) \le l(v) \le l(w)$.
If $l(u) < l(v)$ or $l(v) < l(w)$, then $l(u) < l(w)$ which implies $u \succ w$.

Assume otherwise, i.e., $l(u) = l(v) = l(w)$.
For any two nodes at the same level $s,t$ the distance $d(s,t)$ is odd if and only if the path between them goes through two weak Condorcet winners.
In such a case, we have $s \sim t$. 
Since $u \succ v$ and $v \succ w$ we get that $d(u,v)$ and $d(v,w)$ are even which also implies $d(u,w)$ is even since paths from $v$ to both $u$ and $w$ do not go through two weak Condorcet winners.

It remains to consider the case when $l(u) = l(v) = l(w)$ and there is one Condorcet winner $r$ or two weak Condorcet winners, but all nodes $u,v,w$ are closer to one of them: $r$.

Let $x,y,z$ be the middle nodes of the paths between pairs $(u,v)$, $(v,w)$ and $(u,w)$, respectively.
We will use the fact that the middle node between two nodes $s,t$ from the same level for which $d(s,t)$ is even is the unique node that lays on three path: $s$ to $t$, $s$ to $r$ and $t$ to $r$.
In particular, we know that $x$ and $y$ lay on the path from $v$ to $r$.
Let us consider four cases separately (see Figure~\ref{fig:trees_uvw_transitivity} for an illustration):
\begin{enumerate}
\item[(a)] If $l(x) < l(y)$, then node $x$ is on the path from $u$ to $w$, from $u$ to $r$ and from $w$ to $r$; hence $z = x$. 
From $u \succ v$ we get $|T^u_x| < |T^v_x|$ and from the tree structure $|T^v_x| = |T^w_x|$ which implies $|T^u_x| < |T^w_x|$ and $u \succ w$.

\item[(b)] If $l(x) > l(y)$, then node $y$ is on the path from $u$ to $w$, from $u$ to $r$ and from $w$ to $r$; hence $z = y$.
From $v \succ w$ we get $|T^v_y| < |T^w_y|$ and from the tree structure $|T^u_y| = |T^v_y|$ which implies $|T^u_y| < |T^w_y|$ and $u \succ w$.

\item[(c)] If $x = y \neq z$, then from $u \succ v \succ w$ we get $|T^u_x| < |T^v_x| < |T^w_x|$. 
This is, however, a contradiction as clearly $|T^u_x| = |T^w_x|$.

\item[(d)] If $x = y = z$, then from $u \succ v \succ w$ we get $|T^u_x| < |T^v_x| < |T^w_x|$ which implies $u \succ w$.

\end{enumerate}

This concludes the proof. 
Note that our analysis shows that it is not possible that $x,y,z$ are pairwise different.
\end{proof}

\begin{figure}[t]
\centering
\begin{tikzpicture}[x=1.5cm,y=4cm] 
  \tikzset{     
    e4c node/.style={circle,draw,minimum size=0.50cm,inner sep=0,font=\small}, 
    e4c edge/.style={decorate,decoration={snake,amplitude=0.05cm}},
    triangle/.style={draw,densely dashed}
  }
  
  \def\x{0}
  \def\y{0}
  \node[e4c node] (1) at (\x+-0.2, \y+-0.2) {$r$};
  \node[e4c node] (2) at (\x+-0.4, \y+-0.4) {$x$};
  \node[e4c node] (3) at (\x+-0.0, \y+-0.8) {$u$};
  \node[e4c node] (4) at (\x+-0.6, \y+-0.6) {$y$};
  \node[e4c node] (5) at (\x+-0.4, \y+-0.8) {$w$};
  \node[e4c node] (6) at (\x+-0.8, \y+-0.8) {$v$};

  \node[e4c node,draw=none] (e) at (\x+-0.9, \y+-0.2) {(a)};

  \path[draw,thick]
  (1) edge[e4c edge]  (2)
  (2) edge[e4c edge]  (3)
  (2) edge[e4c edge]  (4)
  (4) edge[e4c edge]  (5)
  (4) edge[e4c edge]  (6)
  ;

  \def\x{2}
  \def\y{0}
  \node[e4c node] (1) at (\x+-0.2, \y+-0.2) {$r$};
  \node[e4c node] (2) at (\x+-0.4, \y+-0.4) {$y$};
  \node[e4c node] (3) at (\x+-0.0, \y+-0.8) {$w$};
  \node[e4c node] (4) at (\x+-0.6, \y+-0.6) {$x$};
  \node[e4c node] (5) at (\x+-0.4, \y+-0.8) {$u$};
  \node[e4c node] (6) at (\x+-0.8, \y+-0.8) {$v$};

  \node[e4c node,draw=none] (e) at (\x+-0.9, \y+-0.2) {(b)};

  \path[draw,thick]
  (1) edge[e4c edge]  (2)
  (2) edge[e4c edge]  (3)
  (2) edge[e4c edge]  (4)
  (4) edge[e4c edge]  (5)
  (4) edge[e4c edge]  (6)
  ;

  \def\x{4}
  \def\y{0}
  \node[e4c node] (1) at (\x+-0.2, \y+-0.2) {$r$};
  \node[e4c node] (2) at (\x+-0.4, \y+-0.4) {$x$};
  \node[e4c node] (3) at (\x+-0.0, \y+-0.8) {$w$};
  \node[e4c node] (4) at (\x+-0.2, \y+-0.6) {$z$};
  \node[e4c node] (5) at (\x+-0.4, \y+-0.8) {$u$};
  \node[e4c node] (6) at (\x+-0.8, \y+-0.8) {$v$};

  \node[e4c node,draw=none] (e) at (\x+-0.9, \y+-0.2) {(c)};

  \path[draw,thick]
  (1) edge[e4c edge]  (2)
  (2) edge[e4c edge]  (4)
  (4) edge[e4c edge]  (3)
  (4) edge[e4c edge]  (5)
  (2) edge[e4c edge]  (6)
  ;

  \def\x{6}
  \def\y{0}
  \node[e4c node] (1) at (\x+-0.2, \y+-0.2) {$r$};
  \node[e4c node] (2) at (\x+-0.4, \y+-0.4) {$x$};
  \node[e4c node] (3) at (\x+-0.0, \y+-0.8) {$w$};
  \node[e4c node] (5) at (\x+-0.4, \y+-0.8) {$u$};
  \node[e4c node] (6) at (\x+-0.8, \y+-0.8) {$v$};

  \node[e4c node,draw=none] (e) at (\x+-0.9, \y+-0.2) {(d)};

  \path[draw,thick]
  (1) edge[e4c edge]  (2)
  (2) edge[e4c edge]  (3)
  (2) edge[e4c edge]  (5)
  (2) edge[e4c edge]  (6)
  ;
\end{tikzpicture}
\caption{An illustration for the proof of Proposition~\ref{proposition:trees_no_cycle}.}
\label{fig:trees_uvw_transitivity}
\end{figure}
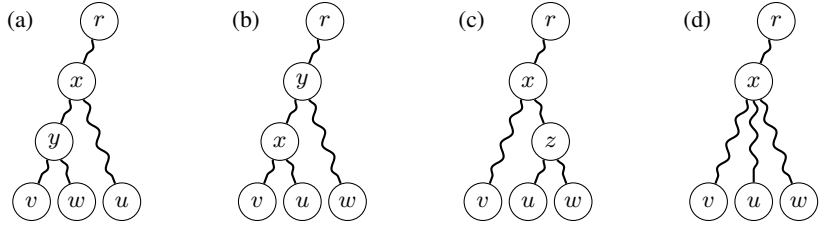

Proposition~\ref{proposition:trees_no_cycle} implies that there is no Condorcet cycle in a tree.
Also, it implies that it is possible to define a centrality measure that satisfies Weak General CCT which we do in the following proposition (see Figure~\ref{fig:trees_new_measure} for an illustration).

\begin{proposition}\label{proposition:trees_measure_exists}
Let $LT^v(G) = (t_1,\dots,t_k)$ be a list of values $\{|T_u^v| : u \in V \setminus \{v\}, |T_u^v| > n/2\}$ sorted increasingly.
A centrality measure $W$ defined for every graph $G = (V,E)$ and node $v$ as follows:
\[ W_v(G) = \left(\sum_{i=1}^k t_i \cdot n^{k+1-i} \right)^{-1} \]
if $k \ge 1$ and $W_v(G) = 1$ if $k=0$ satisfies Weak General CCT.
\end{proposition}
\begin{proof}
The formula for $W_v(G)$ can be written as follows:
\[ W_v^{-1}(G) = t_1 n^k + t_2 n^{k-1} + \dots, t_k n. \]
We know that $n/2 < t_1,\dots,t_k < n$.
Hence, if $|LT^u(G)| < |LT^v(G)|$, then $W_u^{-1}(G) < W_v^{-1}(G)$ which implies $W_u(G)>W_v(G)$.
Now, if $|LT^u(G)| = |LT^v(G)|$ and $i$ is the first position on which both sequences differ, then if $t$-th value of $LT^u(G)$ is smaller, then $W_u^{-1}(G) < W_v^{-1}(G)$; in turn, if $i$-th value of $LT^v(G)$ is smaller, then $W_u^{-1}(G) > W_v^{-1}(G)$.

This can be formalized using the shortlex order: for two sequence of real values, $L_1, L_2$, we say that $L_1 <_{slex} L_2$ if $|L_1|<|L_2|$ or $|L_1|=|L_2|$ and $L_1$ is lexicographically smaller than $L_2$.

So far, we proved that 
\[ LT^u(G) <_{slex} LT^v(G) \Rightarrow W_u(G)>W_v(G). \]
Hence, it remains to prove that 
\[ u \succ v \Rightarrow LT^u(G) <_{slex} LT^v(G). \]

Let us first argue that the level of a node is implied by the set $\{T^v_u\}_{u \in V}$.
If there exists a Condorcet winner $r$, then every subtree rooted at its neighbors must contain less than $n/2$ nodes.
This implies $|T^v_r| > n/2$.
Analogous statement is true for neighbors of weak Condorcet winners, $r,r'$.
Hence, if $v$ is closer to $r$ than to $r'$ we get $|T^v_r| > n/2$ and $|T^v_{r'}| = n/2$.
This implies that if $u$ belongs to the path from $v$ to $r$ (i.e., the Condorcet winner or the closer weak Condorcet winner), then $|T^v_u| > n/2$. 
In turn, if $u$ does not belong to this path, then $|T^v_u| \le n/2$.
Hence, we get that:
\[ l(v) = \{u \in V \setminus \{v\} : |T^v_u| > n/2\}. \]

This implies that if $l(u) < l(v)$, then $LT^u(G) <_{slex} LT^v(G)$.
Also, it implies that if $l(u) = l(v)$, then $|LT^u(G)| = |LT^v(G)|$.
Hence, we need to prove that if $l(u) = l(v)$ and $u \succ v$, then $LT^u(G)$ is lexicographically smaller than $LT^v(G)$.

If $d(u,v)$ is odd, then we know that $u \sim v$.
Hence, we can assume $d(u,v)$ is even, i.e., there exists one Condorcet winner $r$ or two weak Condorcet winners $r,r'$, but $u$ and $v$ are both closer to one of them: $r$.
Let $w$ be the middle node on the path between them. 
Since $u \succ v$ we know that $|T^u_w| < |T^v_w|$.
For every node $w' \neq w$ on the path from $w$ to $r$ value $|T^u_{w'}|=|T^v_{w'}|$ is smaller than both $|T^u_w|$ and $|T^v_w|$.
Also, for every node $w' \neq w$ on the path from $u$ to $w$ we have $|T^u_{w'}| > |T^u_w|$; analogically, for every node $w' \neq w$ on the path from $v$ to $w$ we have $|T^v_{w'}| > |T^v_w|$.
Hence, $|T^u_w|$ and $|T^v_w|$ are the smallest values in $LT^u(G)$ and $LT^v(G)$, respectively, on which both lists differ.
Hence, we get that $LT^u(G) <_{slex} LT^v(G)$.
This concludes the proof.
\end{proof}

We end this section by observing that both Closeness and RW-Closeness centralities depend only on the set $\{T^v_u\}_{u \in V}$.

\begin{proposition}\label{proposition:formulas_based_on_t}
For every tree $G = (V,E)$ and node $v$ it holds:
\[ C_v^{-1}(G) = \sum_{u \in V \setminus \{v\}} |T_u^v|, \quad RWC_v^{-1}(G) = \sum_{u \in V \setminus \{v\}} |T_u^v|(2|T_u^v|-1). \]
\end{proposition}
\begin{proof}
For Closeness centrality, we have:
\[ C_v^{-1}(G) = \sum_{u \in V \setminus \{v\}} |T_u^v| = \sum_{u \in V \setminus \{v\}} \sum_{w \in T^v_u} 1 = \sum_{w \in V \setminus \{v\}} |\{u \in V \setminus \{v\} : w \in T^v_u \}|. \]
Node $w$ belongs to a tree $T^v_u$ if and only if $u$ lays on the path from $w$ to $v$.
Hence, it belongs to $d(w,v)$ such trees and we get the original formula for Closeness centrality.

For RW-Closeness centrality, note that for an edge $\{u,v\}$ if $K(G-\{u,v\}) = \{S_u,S_v\}$, $u \in S_u$ and $v \in S_v$, then $S_u = T^v_u$ and $S_v = T^u_v$.
Hence, from Theorem~\ref{theorem:trees_rwc_consistent} we get that $H(u,v) = 2|T^v_u|-1$.

Now, we have:
\begin{multline*} 
RWC_v^{-1}(G) = \sum_{u \in V \setminus \{v\}} |T_u^v| (2|T^v_u|-1)
= \sum_{u \in V \setminus \{v\}} |T_u^v| H(u,v) \\
= \sum_{u \in V \setminus \{v\}} \sum_{w \in T^v_u} H(u,v) 
= \sum_{w \in V \setminus \{v\}} \sum_{\substack{u \in V \setminus \{v\}\\ w \in T^v_u}} H(u,v).
\end{multline*}

Node $w$ belongs to a tree $T^v_u$ if and only if $u$ lays on the path from $w$ to $v$.
Let $(w,u_1,\dots,u_k,v)$ be this path.
We get:
\[ RWC_v^{-1}(G) = \sum_{w \in V \setminus \{v\}} (H(w,u_1) + \dots + H(u_k,v)) \]
which simplifies to $H(w,v)$ and concludes the proof.
\end{proof}

\begin{figure*}[b]
\centering
\begin{tikzpicture}[x=6cm,y=5.5cm] 
  \tikzset{     
    e4c node/.style={circle,draw,minimum size=0.5cm,inner sep=0}, 
    e4c edge/.style={sloped,above,font=\footnotesize}
  }
  
  \node[e4c node,minimum size=0.58cm] (1) at (0.2, 1.6) {};
 
  \node[e4c node] (6) at (0.185, 1.885) {};
  \node[e4c node] (5) at (0.245, 1.81) {};
  \node[e4c node] (4) at (0.31, 1.74) {};
  \node[e4c node] (1) at (0.2, 1.6) {$u$};
  \node[e4c node] (2) at (0.5, 1.80) {$w$};
  \node[e4c node] (7) at (0.5, 1.70) {};
  \node[e4c node] (8) at (0.5, 1.60) {};
  \node[e4c node] (9) at (0.5, 1.50) {};
  \node[e4c node] (10) at (0.5, 1.40) {};
  \node[e4c node] (3) at (0.8, 1.60) {$v$};
  \node[e4c node] (11) at (0.95, 1.54) {};
  \node[e4c node] (12) at (0.95, 1.66) {};

  \node at (0.2, 1.51) {\small $(8,1,2)$};
  \node at (0.8, 1.51) {\small $(7,4)$};

  \node at (0.8, 1.8) {$G$};

  \path[draw,thick]
  (1) edge[e4c edge]  (2)
  (2) edge[e4c edge]  (3)
  (1) edge[e4c edge]  (4)
  (1) edge[e4c edge]  (5)
  (1) edge[e4c edge]  (6)
  (2) edge[e4c edge]  (4)
  (2) edge[e4c edge]  (5)
  (2) edge[e4c edge]  (6)
  (1) edge[e4c edge]  (7)
  (1) edge[e4c edge]  (8)
  (1) edge[e4c edge]  (9)
  (1) edge[e4c edge]  (10)
  (3) edge[e4c edge]  (12)
  (3) edge[e4c edge]  (7)
  (3) edge[e4c edge]  (8)
  (3) edge[e4c edge]  (9)
  (3) edge[e4c edge]  (10)
  (3) edge[e4c edge]  (11)
  ;
  
  \def\x{1};

  \node[e4c node,minimum size=0.58cm] (1) at (\x+0.8, 1.6) {};

  \node[e4c node] (5) at (\x+0.245, 1.81) {};
  \node[e4c node] (4) at (\x+0.31, 1.74) {};
  \node[e4c node] (1) at (\x+0.2, 1.6) {$u$};
  \node[e4c node] (2) at (\x+0.5, 1.80) {$w$};
  \node[e4c node] (7) at (\x+0.5, 1.70) {};
  \node[e4c node] (8) at (\x+0.5, 1.60) {};
  \node[e4c node] (9) at (\x+0.5, 1.50) {};
  \node[e4c node] (10) at (\x+0.5, 1.40) {};
  \node[e4c node] (3) at (\x+0.8, 1.60) {$v$};
  \node[e4c node] (6) at (\x+0.95, 1.60) {};
  \node[e4c node] (11) at (\x+1.1, 1.54) {};
  \node[e4c node] (12) at (\x+1.1, 1.66) {};

  \node at (\x+0.2, 1.51) {\small $(8,1,2)$};
  \node at (\x+0.8, 1.51) {\small $(7,4)$};

  \node at (\x+0.95, 1.8) {$G'$};

  \path[draw,thick]
  (1) edge[e4c edge]  (2)
  (2) edge[e4c edge]  (3)
  (1) edge[e4c edge]  (4)
  (1) edge[e4c edge]  (5)
  (2) edge[e4c edge]  (4)
  (2) edge[e4c edge]  (5)
  (1) edge[e4c edge]  (7)
  (1) edge[e4c edge]  (8)
  (1) edge[e4c edge]  (9)
  (1) edge[e4c edge]  (10)
  (3) edge[e4c edge]  (7)
  (3) edge[e4c edge]  (8)
  (3) edge[e4c edge]  (9)
  (3) edge[e4c edge]  (10)
  (3) edge[e4c edge] (6)
  (6) edge[e4c edge] (11)
  (6) edge[e4c edge] (12)
  (1) edge[e4c edge, bend left=120, looseness = 1.7] (3)
  ;
\end{tikzpicture}
\caption{Node $u$ and node $v$ have the same lists of distances in both graphs (we list them below the nodes).
Node $v$ has the highest Closeness centrality in both graphs.
However, node $u$ is the Condorcet winner in graph $G$ and node $v$ in graph $G'$.}
\label{fig:graphs_no_distance_based}
\end{figure*}
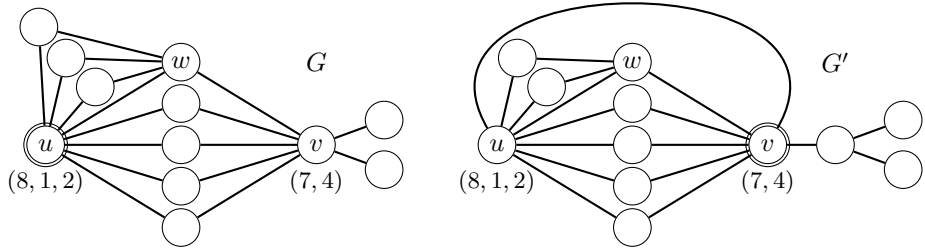

\section{Closeness in General Graphs}

Let us turn our attention to arbitrary graphs.

We start by showing that Closeness centrality satisfies \emph{Condorcet Comparison} also in general graphs.
\begin{definition} (Condorcet Comparison (CC)) \\
For every graph $G = (V,E)$ and edge $\{u,v\} \in E$ it holds:
\[ u \succeq v \Leftrightarrow F_u(G) \ge F_v(G). \]
\end{definition}
The axiom is a direct generalization of Condorcet Comparison in Trees.
Also, it implies the Bridge axiom.

\begin{theorem}\label{theorem:graphs_closeness_satisfies}
Closeness centrality satisfies Condorcet Comparison.
\end{theorem}
\begin{proof}
Since $u$ and $v$ are connected by an edge, for every node $w$ the distance to node $u$ is either one smaller, one larger or equal to the distance to node $v$.
Hence, we get:
\[ C_v^{-1}(G) - C_u^{-1}(G) = Net(u,v) - Net(v,u) \] 
which concludes the proof.
\end{proof}

As before, Condorcet Comparison applies only to nodes $u$ and $v$ connected by an edge.
In general graphs, without such an assumption, the axiom could not be satisfied by any centrality (even in its weakest version: $u \succ v \Rightarrow F_u(G) > F_v(G)$) as implied by the existence of a Condorcet cycle.

As it turns out, Condorcet Comparison does not imply Condorcet consistency on general graphs.
In the following result, we show that Closeness centrality is not Condorcet consistent in general.
Moreover, no other distance-based centrality is Condorcet consistent.

\begin{theorem}
Closeness centrality and all other distance-based centralities are not Condorcet consistent.
\end{theorem}
\begin{proof}
Consider the graph $G$ from Figure~\ref{fig:graphs_no_distance_based}.
It is easy to check that node $u$ is a Condorcet winner.
In turn, node $v$ has the highest Closeness centrality: $C_v^{-1}(G) = 15$ while $C_u^{-1}(G) = 16$ and $C_w^{-1}(G) = 17$.

Now, to show that no distance-based centrality is Condorcet consistent, consider the graph $G'$ from Figure~\ref{fig:graphs_no_distance_based}.
Node $v$ is a Condorcet winner in $G'$.
However, note that nodes $u$ and $v$ both have the same set of distances to other nodes in graph $G$ as in graph $G'$. 
This implies that if $F$ is a distance-based centrality, then $F_u(G) = F_u(G')$ and $F_v(G) = F_v(G')$.
Hence, it is impossible that $\mathrm{Top}_F(G) = \{u\}$ and $\mathrm{Top}_F(G') = \{v\}$.
\end{proof}

Interestingly, RW-Closeness rank first the Condorcet winner in both graphs from Figures~\ref{fig:graphs_no_distance_based}.
However, it is not Condorcet consistent in general, as can be seen in Figure~\ref{fig:preliminaries}.
It is because the random walk is more often in more dense parts of the graph. 
In the graph from Figure~\ref{fig:preliminaries}, it will take more time to go from the left-hand side of the graph to node $v$ than from the right-hand side of the graph to node $u$.
In particular, the random walk requires on average $17$ steps to go from $u$ to $v$, but only $13$ steps to go from $v$ to $u$.
That is why node $u$ is considered more important according to RW-Closeness.

The following proposition formalizes this observation.

\begin{proposition}\label{proposition:graphs_rwc}
RW-Closeness centrality does not satisfy Condorcet Comparison (and Bridge).
Specifically, for every graph $G = (V,E)$ and edge $\{u,v\}$ s.t. $K(G-\{u,v\}) = \{S_u,S_v\}$, $u \in S_u$, $v \in S_v$ it holds:
\[ RWC^{-1}_v(G) - RWC^{-1}_u(G) = |S_u|(2 E[S_u] + 1) - |S_v|(2 E[S_v] + 1). \]
\end{proposition}
\begin{proof}
Fix $\{u,v\} \in E$ and let $K(G-\{u,v\}) = \{S_u, S_v\}$ s.t. $u \in S_u$, $v \in S_v$.
We have:
\[ RWC_v^{-1}(G) =  \sum_{w \in S_u} (H(w,u) + H(u,v)) + \sum_{w \in S_v} H(w,v). \]

Now, let us compute the hitting time $H(u,v)$.
Since $H(u,v)$ concerns the first time the random walk enters node $v$, we can concentrate on the induced subgraph $G[S_u \cup \{v\}]$. 
Consider the expected return time of node $v$ in this subgraph.
Graph $G[S_u \cup \{v\}]$ contains $|E[S_u]|+1$ edges and $v$ is a leaf.
Hence, the expected return time of node $v$ in graph $G[S_u \cup \{v\}]$ equals $2|E[S_u]|+2$.
Now, if the random walk starts in $u$, it will reach node $v$ one step earlier.
Hence, we get that $H(u,v) = 2|E[S_u]|+1$.

Overall, we get that:
\[ RWC_v^{-1}(G) - RWC_u^{-1}(G) = |S_u|(2|E[S_u]|+1) - |S_v|(2|E[S_v]|+1). \]

Now, if $|S_u| < |S_v|$, but $|S_u|(2 E[S_u] + 1) > |S_v|(2 E[S_v] + 1)$, then we have $u \prec v$ and at the same time $RWC_u(G) > RWC_v(G)$.
\end{proof}


Other centrality measures clearly do not satisfy Condorcet Comparison, as they fail to satisfy it already on trees.

In the remainder of this section, we provide a characterization of Closeness centrality ranking based on Condorcet Comparison.
Let us introduce some additional notation concerning lists of distances.
We will say that a list of distances $a = (a_1,\dots,a_k)$ is an $n$-list if $a_1+\dots+a_k = n$.
We define its sum as $S(a) = \sum_{i=1}^k (i a_i)$ and length as $|a| = k$.
For two lists $a = (a_1,\dots,a_k), b = (b_1,\dots,b_l)$ we define $(a+b) = (a_1+b_1,\dots,a_l+b_l,a_{l+1},\dots,a_k)$ if $k \ge l$ and analogously if $k < l$.

First, we introduce a consistency condition on the function $f$ for distance-based centralities.

\begin{definition}
A distance-based centrality based on function $f$ is \emph{regular} if for every two $n$-lists of distances $a,b$ and a list of distances $c$ with $|c| \le |a|,|b|$ it holds:
\[ f(a) \ge f(b) \Leftrightarrow f(a+c) \ge f(b+c). \]
\end{definition}
Regularity states that the comparison between two lists of distances would not change if we add to both lists the same values.
This can be translated to a graph property as comparing one node $v$ in two different graphs $G,G'$.
Now, if a new subgraph is attached to $v$ in both graphs, the same for both, then if $v$ was more central in $G$ than in $G'$, then it will still be more central.

Regularity is satisfied by all standard distance-based centralities, including not yet mentioned \emph{Eccentricity} defined through the function $f(a_1,\dots,a_k) = 1/k$.
However, it is not satisfied by the centrality defined in Equation~\eqref{eq:centrality_zero_to_leafs}.

Now, we characterize the ranking obtained from Closeness centrality among regular distance-based centralities using Condorcet Comparison.

\begin{theorem}\label{theorem:closeness_characterization}
A regular distance-based centrality satisfies Condorcet Comparison if and only if it returns the same ranking as Closeness centrality.
\end{theorem}

Let us start with the key lemma.
Intuitively, this lemma states that if we have a list of distances $a = (a_1,\dots,a_k)$ and two values $a_i,a_j>1$ ($i,j>1$), then we can decrease them by one and increase by one values of $a_{i-1}$ and $a_{j+1}$ without changing the value of the $f$ function.
Similarly, if $a_i > 2$, then we can decrease it by two and increase by one values $a_{i-1}$ and $a_{i+1}$.
An example of $n$-lists of distances $a,b$ that satisfies the lemma assumptions are $a = (3,1,2,4,2)$ and $b = (3,2,1,3,3)$ with $i=3$ and $j=4$.

\begin{lemma}\label{lemma:closeness_shift}
If a regular distance-based centrality based on function $f$ satisfies Condorcet Comparison, then for every two $n$-lists of distances $a = (a_1, \dots, a_k)$, $b = (b_1,\dots,b_m)$ such that there exist $2 \le i \le j \le k$ satisfying:
\[ b_l = a_l - [l=i] - [l=j] + [l \in \{i-1, j+1\}] \]
for every $l \in \{1,\dots,k\}$ it holds $f(a)=f(b)$. Here, we write $[\varphi] = 1$ if $\varphi$ is true, and $[\varphi]=0$, otherwise.
\end{lemma}
\begin{proof}
Fix $2 \le i \le j$.
First, we define two lists of distances $a^*= (a^*_1,\dots,a^*_k)$, $b^* = (b^*_1,\dots,b^*_m)$ satisfying the assumption of the lemma such that $|a^*|=j$ and $\min\{a^*_l,b^*_l\} = 2$ for every $l \in \{1,\dots,j\}$.
Both lists are uniquely characterized by these condition: we have $|b^*|=j+1$ and
\begin{itemize}
\item $a_l = 2 + [l=i] + [l=j]$ for every $l \in \{1,\dots,|a|\}$,
\item $b_l = 2 + [l=i-1] - [l=j+1]$ for every $l \in \{1,\dots,|b|\}$.
\end{itemize}
For example, for $i=2$ and $j=4$ we have $a = (2,3,2,3)$ and $b = (3,2,2,2,1)$.
In turn, if $i=j=3$ we have $a = (2,2,4)$ and $b = (2,3,2,1)$.

Now, consider graph $G^* = (V^*,E^*)$ defined as follows:
\begin{itemize}
\item $V^* = \{u_0,\dots,u_j\} \cup \{v_0,\dots,v_j\} \cup \{w\}$
\item $E^* = \{\{u_l,u_{l+1}\}, \{v_l,v_{l+1}\}: l \in \{0,\dots,j-1\}\} \cup \{\{u_0,v_0\},\{u_{j-1},v_j\},\{v_{i-2}, w\}\}$.
\end{itemize}
See Figure~\ref{fig:closeness_one_left_one_right_appendix} for an illustration.
In this graph we have $A(u_0) = a^*$, $A(v_0) = b^*$ and $u_0 \sim v_0$.
Hence, Condorcet Comparison implies $f(a^*) = f(b^*)$.

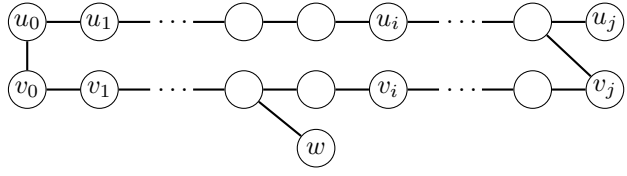
\begin{figure}[t]
\centering
\begin{tikzpicture}[x=8cm,y=6cm] 
  \tikzset{     
    e4c node/.style={circle,draw,minimum size=0.50cm,inner sep=0}, 
    e4c edge/.style={sloped,above,font=\footnotesize}
  }
  \node[e4c node          ] (b1) at (0.00, -0.15) {$v_0$}; 
  \node[e4c node          ] (b2) at (0.12, -0.15) {$v_1$}; 
  \node[e4c node,draw=none] (b3) at (0.24, -0.15) {\ $\dots$}; 
  \node[e4c node          ] (b4) at (0.36, -0.15) {}; 
  \node[e4c node          ] (b5) at (0.48, -0.15) {}; 
  \node[e4c node          ] (bb) at (0.48, -0.28) {$w$}; 
  \node[e4c node          ] (b6) at (0.60, -0.15) {$v_i$}; 
  \node[e4c node,draw=none] (b7) at (0.72, -0.15) {\ $\dots$}; 
  \node[e4c node          ] (b8) at (0.84, -0.15) {}; 
  \node[e4c node          ] (b9) at (0.96, -0.15) {$v_j$}; 


  \node[e4c node          ] (t1) at (0.00, 0.00) {$u_0$}; 
  \node[e4c node          ] (t2) at (0.12, 0.00) {$u_1$}; 
  \node[e4c node,draw=none] (t3) at (0.24, 0.00) {\ $\dots$}; 
  \node[e4c node          ] (t4) at (0.36, 0.00) {}; 
  \node[e4c node          ] (t5) at (0.48, 0.00) {}; 
  \node[e4c node          ] (t6) at (0.60, 0.00) {$u_i$}; 
  \node[e4c node,draw=none] (t7) at (0.72, 0.00) {\ $\dots$}; 
  \node[e4c node          ] (t8) at (0.84, 0.00) {}; 
  \node[e4c node          ] (t9) at (0.96, 0.00) {$u_j$}; 

  \path[draw,thick]
  (b1) edge[e4c edge] (b2)
  (b2) edge[e4c edge] (b3)
  (b3) edge[e4c edge] (b4)
  (b4) edge[e4c edge] (b5)
  (b5) edge[e4c edge] (b6)
  (b6) edge[e4c edge] (b7)
  (b7) edge[e4c edge] (b8)
  (b8) edge[e4c edge] (b9)

  (t1) edge[e4c edge] (t2)
  (t2) edge[e4c edge] (t3)
  (t3) edge[e4c edge] (t4)
  (t4) edge[e4c edge] (t5)
  (t5) edge[e4c edge] (t6)
  (t6) edge[e4c edge] (t7)
  (t7) edge[e4c edge] (t8)
  (t8) edge[e4c edge] (t9)

  (b1) edge[e4c edge] (t1)
  (b4) edge[e4c edge] (bb)
  (t8) edge[e4c edge] (b9)
  ;
  
\end{tikzpicture}
\caption{Graph $G^*$ from the proof of Lemma~\ref{lemma:closeness_shift}.
We have $u_0 \sim v_0$,\break $A(u_0) = (2,\dots,2,3,\dots,3)$ and $A(v_0) = (2,\dots,3,2,\dots,2,1)$.}
\label{fig:closeness_one_left_one_right_appendix}
\end{figure}

Now, let us discuss how to generalize this result on all possible $n$-lists of distances.
Take two arbitrary $n$-lists of distances $a = (a_1,\dots,a_k)$, $b = (b_1,\dots,b_m)$ that satisfy the lemma restrictions, i.e.,
\[ b_l = a_l - [l=i] - [l=j] + [l \in \{i-1,j+1\}] \]
for every $l \in \{1,\dots,k\}$.
Assume $\min \{a_l,b_l\} \ge 2$ for every $l \in \{1,\dots,j\}$.
Since $b_l - a_l = b^*_l-a^*_l$, we have $a_l-a^*_l = b_l-b^*_l \ge 0$ for every $l \in \{1,\dots,k\}$.
Hence, we extend graph $G^*$ as follows:
\begin{itemize}
\item For every $l \in \{1,\dots,j\}$ we add $a_l-a^*_l$ nodes connected by two edges with $u_{l-1}$ and $v_{l-1}$.
In this way, these nodes are at distance $l$ from both $u_0$ and $v_0$.
\item To accommodate for $l > j$, take any graph $G' = (V',E')$ with $s \in V'$ such that $A(s) = (a_{j+1}, \dots, a_k)$.
Now, we add this graph to $G^*$ and merge node $s$ from $G'$ with node $v_j$ from $G^*$.
In this way, for every $l > j$, we add $a_l$ nodes at distance $l$ from both $u_0$ and $v_0$.
\end{itemize}
See Figure~\ref{fig:closeness_one_left_one_right_appendix_2} for an example.
Now, in the resulting graph, we have $A(u_0) = a$ and $A(v_0) = b$.
Hence, Condorcet Comparison implies $f(a) = f(b)$.

\begin{figure}[b]
\centering
\begin{tikzpicture}[x=7cm,y=6cm] 
  \tikzset{     
    e4c node/.style={circle,draw,minimum size=0.50cm,inner sep=0}, 
    e4c edge/.style={sloped,above,font=\footnotesize}
  }
  
  \def\y{0.05}
  
  \node[e4c node          ] (t0) at (0.00, 0.00) {$u_0$}; 
  \node[e4c node          ] (t1) at (0.12, 0.00) {$u_1$}; 
  \node[e4c node          ] (t2) at (0.24, 0.00) {$u_2$}; 
  \node[e4c node          ] (t3) at (0.36, 0.00) {$u_3$}; 
  \node[e4c node          ] (t4) at (0.48, 0.00) {$u_4$}; 

  \node[e4c node          ] (b0) at (0.00, -6*\y) {$v_0$}; 
  \node[e4c node          ] (b1) at (0.12, -6*\y) {$v_1$}; 
  \node[e4c node          ] (b2) at (0.24, -6*\y) {$v_2$}; 
  \node[e4c node          ] (b3) at (0.36, -6*\y) {$v_3$}; 
  \node[e4c node          ] (b4) at (0.48, -6*\y) {$v_4$}; 
  
  \node[e4c node          ] (bb) at (0.12, -6*\y-0.12) {$w$}; 
  
  \node[e4c node          ] (m1) at (0.12, -3*\y) {};
  \node[e4c node          ] (m2) at (0.24, -2*\y) {};
  \node[e4c node          ] (m3) at (0.24, -4*\y) {};
  \node[e4c node          ] (m4) at (0.60, -2*\y) {};
  \node[e4c node          ] (m5) at (0.60, -4*\y) {};

  \path[draw,thick]
  (b0) edge[e4c edge] (b1)
  (b1) edge[e4c edge] (b2)
  (b2) edge[e4c edge] (b3)
  (b3) edge[e4c edge] (b4)

  (t0) edge[e4c edge] (t1)
  (t1) edge[e4c edge] (t2)
  (t2) edge[e4c edge] (t3)
  (t3) edge[e4c edge] (t4)

  (b0) edge[e4c edge] (t0)
  (b0) edge[e4c edge] (bb)
  (t3) edge[e4c edge] (b4)

  (t0) edge[e4c edge] (m1)
  (b0) edge[e4c edge] (m1)
  
  (t1) edge[e4c edge] (m2)
  (b1) edge[e4c edge] (m2)
  
  (t1) edge[e4c edge] (m3)
  (b1) edge[e4c edge] (m3)
  
  (b4) edge[e4c edge] (m4)
  (b4) edge[e4c edge] (m5)
  ;
  
\end{tikzpicture}
\caption{Modified graph $G^*$ from the proof of Lemma~\ref{lemma:closeness_shift}. We have $u_0 \sim v_0$, $A(u_0) = (3,5,2,3,2)$ and $A(v_0) = (4,4,2,2,3)$.}
\label{fig:closeness_one_left_one_right_appendix_2}
\end{figure}
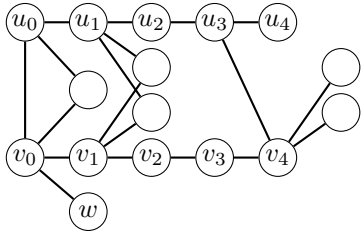

Finally, assume the condition $\min \{a_l,b_l\} \ge 2$ does not hold for every $l \in \{1,\dots,j\}$.
Here, we use regularity of function $f$.
Let $c = (1,\dots,1)$ be a $j$-list of distances.
Now, $a+c$ and $b+c$ are two $(n+j)$-lists of distances that satisfy the lemma assumptions and also $\min \{(a+c)_l,(b+c)_l\} \ge 2$ for every $l \in \{1,\dots,j\}$.
Hence, $f(a+c) = f(b+c)$ which from regularity of $f$ implies $f(a) = f(b)$.
\end{proof}

Let us define the \emph{weight} of a list of distances $a = (a_1,\dots,a_k)$ as follows:
\[ \omega(a) = \sum_{i=2}^{|a|} (a_i-1). \]
Clearly, weight of a list is non-negative.
Now, we will show that for every $n$ and every possible sum $S$ (every possible sum of an $n$-list of distances) there exists exactly one $n$-list with weight $0$ or $1$; we will denote it by $\bot_{S,n}$.

\begin{lemma}\label{lemma:closeness_minimal_exists}
For every $n \le S \le \frac{n(n+1)}{2}$ there exists a unique $n$-list of distances $\bot_{S,n}$ with the sum $S(\bot_{S,n}) = S$ and $\omega(\bot_{S,n}) \le 1$.
\end{lemma}
\begin{proof}
If $\omega(a) \le 1$, then this means that there exists at most one $j \in \{2,\dots,k\}$ such that $a_j > 1$.
If such $j$ does not exist, assume $j=1$.
We get:
\[ S = \sum_{i=1}^k ia_i = n+(1+\dots+k-1)+(j-1). \]
Since $j \le k$ in this formula, we get that $k$ is uniquely defined as the minimal natural number s.t. $(\sum_{i=1}^k i) > (S-n)$ and $j = (S-n) - (\sum_{i=1}^{k-1} i) + 1$.
\end{proof}
For example, for $S = 28$ and $n=11$ we have $k=6$ and $j=3$, hence $\bot_{28,11} = (5,1,2,1,1,1)$.

Based on Lemma~\ref{lemma:closeness_shift} we can show for every $n$-list of distances $a$ it holds $f(a) = f(\bot_{S(a),n})$.

\begin{lemma}\label{lemma:closeness_equal_minimal}
If a regular distance-based centrality based on function $f$ satisfies Condorcet Comparison, then for every $n$-list of distances $a = (a_1, \dots, a_k)$ it holds:
\[ f(a) = f(\bot_{S(a),n}). \]
\end{lemma}
\begin{proof}
We proceed by induction on the weight of a list of distances.
Take an $n$-list of distances $a = (a_1,\dots,a_k)$ and assume $\omega(a) > 1$ which means there exist $a_i,a_j > 1$ such that $i < j$ or $a_i > 2$ (in such a case, let $j=i$).

Let $m = \min \{i-1, k-j+1\}$.
If $i-1 \ge k-j+1$, then by using Lemma~\ref{lemma:closeness_shift} $m$ times we get $f(a) = f(b)$ for $b = (b_1,b_2,\dots,b_k,1)$ defined as follows:
\[ \forall_{l \in \{1,\dots,k\}} \left(b_l = a_l - [l=i] - [l=j] + [l = i-m] \right) \]
In turn, if $i-1 < k-j+1$, then using Lemma~\ref{lemma:closeness_shift} $m$ times we get $f(a) = f(b)$ for $b = (a_1+1,b_2,\dots,b_k)$ defined as follows:
\[ \forall_{l \in \{2,\dots,k\}} \left(b_l = a_l - [l=i] - [l=j] + [l = j+m] \right) \]
In both cases, we get that $\omega(b) < \omega(a)$.

Eventually, we get the list of distances with the weight smaller or equal to one. Lemma~\ref{lemma:closeness_minimal_exists} implies this list is $\bot_{S(a),n}$.
\end{proof}

To give an example, for $a = (4,1,2,4)$, if we always choose the smallest possible $i$ and the largest possible $j$ we get $(4,2,1,3,1)$, $(5,1,1,2,2)$ and eventually\break $\bot_{28,11} = (5,1,2,1,1,1)$.

The last ingredient of the proof is lemma that states $\bot_{S,n} < \bot_{S+1,n}$ for every sum $S$.
\begin{lemma}\label{lemma:closeness_two_minimal}
If a regular distance-centrality based on function $f$ satisfies Condorcet Comparison, then for every $n \le S < \frac{n(n+1)}{2}$ it holds $f(\bot_{S,n}) > f(\bot_{S+1,n})$.
\end{lemma} 
\begin{proof}
Let $\bot_{S,n} = (a_1,\dots,a_k)$.
We know that $\bot_{S,n}$ is of the form $(m,1,\dots,1,2,1,\dots,1)$, or $(m,1,\dots,1)$. 
Let $j$ be the index such that $a_j > 1$ for $j > 1$ (or $j=1$ if such index does not exist).
Note that lists $\bot_{S,n}$ and $\bot_{S+1,n}$ may differ only on position $j$ and $j+1$.

Let $c = (1,\dots,1)$ be a $j$-list of distances
We will now create a graph in which two adjacent nodes have lists of distances $\bot_{S,n}+c$ and $\bot_{S+1,n}+c$.

To this end, consider graph $G^* = (V^*,E^*)$ defined as follows:
\begin{itemize}
\item $V^* = \{u_0,\dots,u_k\} \cup \{v_0,\dots,v_j\} \cup \{w\} \cup \{z_1,\dots,z_{a_1-1}\}$
\item $E^* = \{\{u_l,u_{l+1}\}, \{v_l,v_{l+1}\}: l \in \{0,\dots,j-1\}\} 
\cup \{\{y_l,y_{l+1}\} : l \in \{1,\dots,j-k-1\}\}
\cup \{\{z_l,u_0\}, \{z_l,v_0\} : l \in \{1,\dots,a_1-1\}\}
\cup \{\{u_0,v_0\},\{u_{j-1},v_j\},\{v_{i-2}, w\}\}$.
\end{itemize}
See Figure~\ref{fig:closeness_two_with_minimal_appendix} for an illustration.
Node $u_0$ has $a_1+1$ nodes at distance 1 (all $z$'s, $v_0$ and $v_1$), $3$ nodes at distance $j$ ($u_j$, $w$ and $v_{j-1}$), $2$ nodes at distance $1 < l < j$ ($u_l$ and $v_{l-1}$) and $1$ nodes at distance $l > j$ ($u_l$).
Hence, we have $A(u_0) = \bot_{S,n}+c$.
Now, node $v_0$ has the same list of distances with the only difference in the distance to node $w$.
Hence, $A(v_0) = \bot_{S+1,n}+c$.

\begin{figure}[t]
\centering
\begin{tikzpicture}[x=7cm,y=6cm] 
  \tikzset{     
    e4c node/.style={circle,draw,minimum size=0.50cm,inner sep=0}, 
    e4c edge/.style={sloped,above,font=\footnotesize}
  }

  \node[e4c node          ] (l1) at (-0.12, +0.05) {$z_1$}; 
  \node[e4c node          ] (l2) at (-0.12, -0.075) {$z_2$}; 
  \node[e4c node,draw=none] (l3) at (-0.12, -0.20) {\ $\dots$}; 

  \node[e4c node          ] (bb) at (0.48, +0.075) {$w$}; 

  \node[e4c node          ] (t1) at (0.00, 0.00) {$u_0$}; 
  \node[e4c node          ] (t2) at (0.12, 0.00) {$u_1$}; 
  \node[e4c node,draw=none] (t3) at (0.24, 0.00) {\ $\dots$}; 
  \node[e4c node          ] (t4) at (0.36, 0.00) {}; 

  \node[e4c node          ] (b1) at (0.00, -0.15) {$v_0$}; 
  \node[e4c node          ] (b2) at (0.12, -0.15) {$v_1$}; 
  \node[e4c node,draw=none] (b3) at (0.24, -0.15) {\ $\dots$}; 
  \node[e4c node          ] (b4) at (0.36, -0.15) {}; 
  
  \node[e4c node          ] (t5) at (0.48, -0.075) {$u_j$}; 
  \node[e4c node,draw=none] (t6) at (0.60, -0.075) {\ $\dots$}; 
  \node[e4c node          ] (t7) at (0.72, -0.075) {$u_k$}; 

  \path[draw,thick]
  (t1) edge[e4c edge] (t2)
  (t2) edge[e4c edge] (t3)
  (t3) edge[e4c edge] (t4)
  (t4) edge[e4c edge] (t5)
  (t5) edge[e4c edge] (t6)
  (t6) edge[e4c edge] (t7)

  (b1) edge[e4c edge] (b2)
  (b2) edge[e4c edge] (b3)
  (b3) edge[e4c edge] (b4)

  (b1) edge[e4c edge] (t1)
  (t4) edge[e4c edge] (bb)
  (b4) edge[e4c edge] (t5)
  
  (l1) edge[e4c edge] (b1)
  (l1) edge[e4c edge] (t1)
  (l2) edge[e4c edge] (b1)
  (l2) edge[e4c edge] (t1)
  (l3) edge[e4c edge] (b1)
  (l3) edge[e4c edge] (t1)
  ;
  
\end{tikzpicture}
\caption{Graph $G^*$ from the proof of Lemma~\ref{lemma:closeness_two_minimal}.
We have $u_0 \sim v_0$,\break $A(u_0) = (2,\dots,2,3,\dots,3)$ and $A(v_0) = (2,\dots,3,2,\dots,2,1)$.}
\label{fig:closeness_two_with_minimal_appendix}
\end{figure}
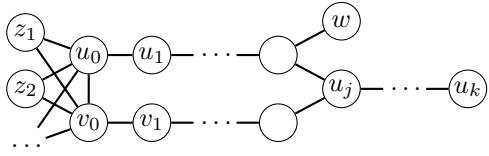

From Condorcet Comparison and the fact that $u \succ v$ we get $f(\bot_{S,n}+c) > f(\bot_{S+1,n}+c)$.
This, based on regularity, implies $f(\bot_{S,n}) > f(\bot_{S+1,n})$.
\end{proof}

We are now ready to present the proof of Theorem~\ref{theorem:closeness_characterization} based on the above lemmas.

\begin{proof}[Proof of Theorem~\ref{theorem:closeness_characterization}]
The ``if'' part follows from Theorem~\ref{theorem:graphs_closeness_satisfies}: Closeness centrality satisfies Condorcet Comparison, hence every centrality $F$ that returns the same ranking also satisfies it: for every graph $G = (V,E)$ and edge $\{u,v\} \in E$ we have:
\[ F_u(G) \ge F_v(G) \Leftrightarrow C_u(G) \ge C_v(G) \Leftrightarrow u \succeq v. \]

Let us focus on the ``only if'' part.
Let $F$ be a regular distance-based centrality based on function $f$ that satisfies Condorcet Comparison.
Take arbitrary graph $G$ and two nodes $u,v$.
We will show that
\[ F_u(G) \ge F_v(G) \Leftrightarrow C_u(G) \ge C_v(G). \]
Since $F$ is a distance-based centrality and $C_u(G) = (S(A(u)))^{-1}$ this is equivalent to showing
\[ f(A(u)) \ge f(A(v)) \Leftrightarrow S(A(u)) \le S(A(v)). \]
Hence, in what follows, we will show that for any two $n$-lists of distances it holds:
\begin{equation}\label{eq:theorem_closeness_char}
f(a) \ge f(b) \Leftrightarrow S(a) \le S(b).
\end{equation}

From Lemma~\ref{lemma:closeness_equal_minimal} we get that $f(a) = f(\bot_{S(a),n})$ and $f(b) = f(\bot_{S(b),n})$.
If $S(a) = S(b)$, then this implies $f(a) = f(b)$ and Equation~\eqref{eq:theorem_closeness_char} is satisfied.

Assume otherwise and without loss of generality let $S(a) < S(b)$.
From Lemma~\ref{lemma:closeness_two_minimal} we get that $f(\bot_{S(a),n}) > f(\bot_{S(a)+1,n}) > \dots > f(\bot_{S(b),n})$ and Equation~\eqref{eq:theorem_closeness_char} is satisfied. 
This concludes the proof.
\end{proof}

\section{Conclusions}
We studied the connection between Closeness centrality and the Condorcet principle. 
We showed that on trees, Closeness centrality, along with its variant Random-Walk Closeness, are the only popular centrality measure that is Condorcet consistent.
In general graphs, no distance-based centrality is Condorcet consistent, but Closeness centrality is the only regular distance-based centrality that satisfies Condorcet Comparison: an axiom that states that out of two adjacent nodes, the one with more nodes closer to it has the higher centrality.

There are several potential future directions of this work.
First of all, it would be interesting to identify sufficient and necessary conditions for a Condorcet winner to exist in general graphs. 
Also, if a Condorcet winner does not exist, then the structure of the \emph{top cycle} can be analyzed.
In particular, our results imply that a Condorcet cycle cannot be a cycle in a graph, as we know that along the path Closeness centrality always decreases.
Furthermore, larger classes of graphs on which Closeness (or RW-Closeness) is Condorcet consistent can be characterized.
Another idea is to modify the model and treat the edges, not the nodes, as voters.
This is inspired by the analysis of RW-Closeness that leans toward parts of the graph with more edges but not necessarily more nodes.

Yet another idea is to study weaker Condorcet criteria such as the Condorcet loser criterion that, in terms of graphs, would mean that a node that is losing to all other nodes in a head-to-head comparison has the highest centrality.
Another weaker criterion could be a local version of Condorcet consistency that states nodes with the highest centrality win in a head-to-head comparison with their neighbors. 
Condorcet Comparison implies that this property is, in fact, satisfied by Closeness centrality.
On top of that, other concepts from social choice theory could be analyzed.
Finally, the setting could be extended to directed graphs.

\section*{Acknowledgments}
This work was created as a part of the project ``Ordinal centrality measures'' at University of Warsaw.

\bibliographystyle{plainnat}
\bibliography{bibliography}

\end{document}